\documentclass[12pt,reqno]{amsart}

\usepackage[margin=1in]{geometry}
\usepackage[english]{babel}
\usepackage[utf8]{inputenc}
\usepackage{fancyhdr}
\usepackage{natbib}
\usepackage{amsmath}

\usepackage{algorithm}
\usepackage[noend]{algpseudocode}

\usepackage{xcolor}

\RequirePackage{amsthm,amsmath,amsfonts,amssymb}
\usepackage{graphicx}
\usepackage{graphicx}
\usepackage{placeins}
\usepackage{caption}
\usepackage{subcaption}
\numberwithin{equation}{section}

\usepackage{adjustbox}
\usepackage{rotating}

\usepackage{multirow}

\usepackage{hyperref}

\newtheorem{theorem}{Theorem}[section]

\newtheorem{lemma}{Lemma}[section]

\newtheorem{remark}[theorem]{Remark}%[section]

\theoremstyle{definition}
\RequirePackage{amsthm,amsmath,amsfonts,amssymb}
\usepackage{graphicx}
\usepackage{lscape}
\usepackage{placeins}
\usepackage{caption}
\usepackage{subcaption}
\numberwithin{equation}{section}

\theoremstyle{plain}
\theoremstyle{remark}
\numberwithin{equation}{section}

\usepackage{multirow,setspace}

\title{Jackknife Empirical Likelihood for Multivariate three-Sample U-Statistics and its Applications}

\author{N\lowercase{aresh} G\lowercase{arg$^a$}, L\lowercase{itty} M\lowercase{athew$^b$}, I\lowercase{sha} D\lowercase{ewan$^a$ and} S\lowercase{udheesh} K\lowercase{umar} K\lowercase{attumannil$^{c,\dag}$}\\
\lowercase{$^a$} S\lowercase{tatistical} S\lowercase{ciences} U\lowercase{nit},
	I\lowercase{ndian} S\lowercase{tatistical} I\lowercase{nstitute}, D\lowercase{elhi}, I\lowercase{ndia},\\ %\texttt{garg.naresh22@gmail.com, ishadewan@gmail.com
\lowercase{$^b$}S\lowercase{chool of} C\lowercase{omputer} S\lowercase{cience and} S\lowercase{tatistics},
	T\lowercase{rinity} C\lowercase{ollege} D\lowercase{ublin}, D\lowercase{ublin}, I\lowercase{reland, }\\%\texttt{mathewli@tcd.ie}}
\lowercase{$^c$}S\lowercase{tatistical} S\lowercase{ciences} U\lowercase{nit},
	I\lowercase{ndian} S\lowercase{tatistical} I\lowercase{nstitute}, C\lowercase{hennai}, I\lowercase{ndia}.}

\thanks{$\dag$Corresponding Author mail: skkattu@isichennai.res.in}

\begin{document}
	\maketitle

%\doublespace
\begin{abstract}
We develop a jackknife empirical likelihood (JEL) framework for inference on parameters defined through multivariate three-sample U-statistic. From three independent multivariate samples, we construct JEL ratio statistic based on suitable jackknife pseudo-values and, under mild regularity conditions, establish a Wilks-type result showing that the log JEL ratio converges in distribution to a chi-square limit. This provides asymptotically valid confidence intervals for the parameter of interest without explicit variance estimation or heavy resampling. To illustrate the usefulness of the proposed method, we construct confidence intervals for differences in volume under the surface (VUS) measures, which are widely used in classification problems. Through Monte Carlo simulations, we compare the performance of JEL-based confidence intervals with those obtained from normal approximation of U-statistic and kernel-based methods. The findings indicate that the proposed JEL approach outperforms existing methods in terms of coverage probability and computational efficiency. Finally, we apply our methods to a recent real dataset. 
\\

\noindent \textbf{Keyword:}
	Jackknife empirical likelihood (JEL); Volume under surface (VUS);	U-statistics.
\end{abstract}

	\section{\textbf{Introduction}}
Empirical likelihood (EL) has emerged as a useful tool in statistical inference when traditional parametric assumptions about the distribution of the data are either unknown or do not hold. EL operates on the principle of maximizing a nonparametric likelihood function with constraints  on  observed data with applications to  hypotheses testing and  confidence interval estimation. Early work on EL is due to Owen (\citeyear{owen1988empirical}, \citeyear{owen1990empirical}). \cite{owen2001empirical} highlighted that empirical likelihood (EL) inherently constructs confidence regions, accommodates supplementary constraints, and preserves the natural scale and transformations of the data. This adaptability has established EL as a useful methodology in  in handling incomplete data and addressing complex multivariate problems (\cite{jian2008}).\\

%{\color{green} Last two sentences are not clear. Can we simplify them?}

    For nonlinear statistics such as U-statistics, however, the original EL formulation can be difficult to implement directly, because the constraints involve nonlinear functionals of the empirical distribution. To address this difficulty, \cite{jing2009jackknife} proposed the jackknife empirical likelihood (JEL) method, which constructs empirical likelihood ratios based on jackknife pseudo-values. Their work established Wilks-type results for one- and two-sample univariate U-statistics and demonstrated that JEL often improves coverage accuracy over normal approximation methods while avoiding explicit variance estimation. Since then, JEL has been extended to various settings, including multivariate mean comparisons \cite{Wang2013}, non-smooth estimating equations \cite{Li2016}, and tests based on energy statistics \cite{SANG2020}, underscoring its flexibility for complex nonparametric functionals. \cite{sang2021jackknife} developed a jackknife empirical likelihood procedure for nonparametric $k$-sample tests based on categorical Gini correlations, and established a chi-square limiting distribution for the resulting test statistic.
\\

Despite this progress, existing JEL developments for U-statistics are predominantly limited to one- or two-sample problems, or to vector parameters in a single-sample setting. In many applications, however, the parameter of interest is naturally expressed as a functional of multivariate three-sample U-statistics. Examples include ordering probabilities in three-arm clinical trials, multivariate stress--strength reliability problems, and measures of three-class separation in diagnostic and classification settings. In such problems, the underlying U-statistic is built from three independent multivariate samples, and the usual one- or two-sample JEL theory is not directly applicable. A general JEL framework for multivariate $3$-sample U-statistics is therefore needed.
\\

The goal of this paper is to develop such a framework. We consider U-statistics built from three independent $q$-dimensional samples with degrees $(m_1,m_2,m_3)$ and a symmetric kernel, and define a combined-sample construction that allows us to work with jackknife pseudo-values that are asymptotically independent but not identically distributed. Within this setting, we obtain a jackknife empirical likelihood ratio for a general scalar parameter $\theta$ defined through a three-sample U-statistic and establish a Wilks-type theorem: under mild moment and sample-size conditions, the JEL log-likelihood ratio converges in distribution to a chi-square random variable with one degree of freedom. This yields asymptotically valid confidence intervals for $\theta$ without explicit variance estimation or resampling.
\\

To illustrate the practical relevance of the proposed methodology, we focus on problems arising in diagnostic accuracy assessment. For three-class classification, we construct confidence intervals for the difference between two correlated volumes under ROC surfaces (VUSs). In this case, the parameters of interest can be written as functionals of multivariate three-sample U-statistics, and our general JEL theory provides a coherent nonparametric inference procedure. Extensive Monte Carlo simulations under several multivariate models, including the Marshall--Olkin bivariate exponential and FGM-based bivariate Pareto distributions, are used to study the finite-sample performance of the proposed JEL intervals in comparison with intervals based on asymptotic normality with jackknife variance estimation and kernel-based bootstrap methods. Finally, we apply the methods to biomarker data from the Alzheimer's Disease Neuroimaging Initiative, where we construct confidence intervals for differences in VUSs when comparing diagnostic markers across cognitive groups.
\\

The  paper is structured as follows. Section 2 we study $q$-dimensional U-statistics based on three samples. We study its asymptotic properties and find a consistent estimator of the unknown variance of the U-statistic based  on Jackknife pseudo values. Jackknife empirical likelihood approach for these U-statistics is also discussed.  As an illustration,in Section 3, we find confidence intervals for the difference between two dependent volumes under surfaces (VUS) using  asymptotic normality of U-statistics, JEL method  and bootstrapping of kernel type estimators . Section 4 presents extensive simulation studies assessing the performance of these confidence intervals under the Marshall–Olkin bivariate exponential and Bivariate Pareto models. Section 4 illustrates the methodology through a real-data analysis. Section 5 highlights the findings of the study along with some open problems. 

 \section{\textbf{Three sample multivariate  U-Statistics  }} \label{Theory}
Suppose $\left(\bold{X}_{1}, \ldots, \bold{X}_{n_{1}}\right),\left(\bold{Y}_{1}, \ldots, \bold{Y}_{n_{2}}\right)$ and $\left(\bold{Z}_{1}, \ldots, \bold{Z}_{n_{3}}\right)$ represents three independent random samples from $q$-dimensional distribution functions $F$, $G$ and $H$, respectively, where $\bold{X}_{i}=(X_{i1},X_{i2},\dots,X_{iq})^T,\; i=1,2,...,n_1$, $\bold{Y}_{j}=(Y_{j1},Y_{j2},\dots,Y_{jq})^T,\;$ $ j=1,2,...,n_2,$ and $\bold{Z}_{k}=(Z_{k1},Z_{k2},\dots,Z_{kq})^T,\; k=1,2,...,n_3$, and $q\geq 1$. \\

Let $$\theta=E[h\left(\bold{X}_{i_{1}}, \ldots, \bold{X}_{i_{m_{1}}}, \bold{Y}_{j_{1}}, \ldots, \bold{Y}_{j_{m_{2}}}, \bold{Z}_{k_{1}} \ldots, \bold{Z}_{k_{m_{3}}}\right)]\in \Re$$ be the parameteric function  of interest, where 
$h : \mathcal{X}^{m_{1}} \times \mathcal{Y}^{m_{2}} \times \mathcal{Z}^{m_{3}} \;\rightarrow\; \Re$
is a measurable  function symmetric in its arguments, i.e., invariant under any permutation of indices within each sample, and  $\mathcal{X}, \mathcal{Y},$ and $ \mathcal{Z}$ be the sample spaces associated with the distribution functions $F$, $G$, and $H$, respectively.

An unbiased estimator of $\theta$  based on  three-sample U-statistic with a symmetric kernel $h$  of degree $\left(m_{1}, m_{2}, m_{3}\right)$ is given by
\begin{align} \label{ustat}
	U_{n_{1}, n_{2}, n_{3}}= & \frac{1}{{n_1\choose m_1}{n_2\choose m_2}{n_3\choose m_3}}  \mathop{\sum\sum\sum}_{\substack{1 \leq i_{1}<\ldots<i_{m_{1}} \leq n_{1} \\ 1 \leq j_{1}<\ldots<j_{m_{2}} \leq n_{2} \\ 1 \leq k_{1}<\ldots<k_{m_{3}} \leq n_{3} }} h\left(\bold{X}_{i_{1}}, \ldots, \bold{X}_{i_{m_{1}}}, \bold{Y}_{j_{1}}, \ldots, \bold{Y}_{j_{m_{2}}}, \bold{Z}_{k_{1}} \ldots, \bold{Z}_{k_{m_{3}}}\right),
\end{align}
where $\{i_{1}, \ldots, i_{m_1}\}$, $\{j_{1}, \ldots, j_{m_2}\}$ and    $\{k_{1}, \ldots, k_{m_3}\}$ indicates $m_1$, $m_2$ and $m_3$ integers chosen without replacement from sets $\{1, \ldots, n_1\}$, $\{1, \ldots, n_2\}$ and $\{1, \ldots, n_3\}$, respectively.

\subsection{Three sample multivariate  U-Statistic -Asymptotic results}
\label{CI for U statistics}
\vspace{.01in}

In this section  we study the asymptotic properties of the U-statistic \(U_{n_{1}, n_{2}, n_{3}}\). Define

\begin{align*}
	\sigma_{100}^{2}&=Var\left[E\left(h\left(\bold{X}_{1},\dots,\bold{X}_{m_1}, \bold{Y}_{1},\dots,\bold{Y}_{m_2}, \bold{Z}_{1},\dots, \bold{Z}_{m_3}\right) \mid \bold{X}_1\right)\right], \\
	\sigma_{010}^{2}&=Var\left[E\left(h\left(\bold{X}_{1},\dots,\bold{X}_{m_1}, \bold{Y}_{1},\dots,\bold{Y}_{m_2}, \bold{Z}_{1},\dots, \bold{Z}_{m_3}\right) \mid \bold{Y}_1\right)\right]
\end{align*}
\text{and}
\begin{align*}
	\sigma_{001}^{2}&=Var\left[E\left(h\left(\bold{X}_{1},\dots,\bold{X}_{m_1}, \bold{Y}_{1},\dots,\bold{Y}_{m_2}, \bold{Z}_{1},\dots, \bold{Z}_{m_3}\right) \mid \bold{Z}_1\right)\right].
\end{align*}

\vspace{.05in}

 \begin{theorem} \label{lemma1}
 
	\vspace{.1in}
    
    \noindent(a) $U_{n_{1}, n_{2}, n_{3}} \xrightarrow{\text { a.s }} \theta$ as $\min\{n_1,n_2,n_3\} \rightarrow \infty$;
    
    \vspace{.1in}
    
	\noindent(b) Suppose that $\sigma_{100}^{2}, \sigma_{010}^{2}, \sigma_{001}^{2}>0$  holds. Let
	$\sigma^2=\frac{m_1^2\sigma_{100}^{2}}{ n_{1}}+\frac{m_2^2\sigma_{010}^{2}}{n_{2}}+\frac{m_3^2\sigma_{001}^{2}}{ n_{3}},$
	Then,
	\begin{equation*}
		\frac{U_{n_{1}, n_{2}, n_{3}}-\theta}{\sigma} \xrightarrow{d} N(0,1), \text { as } \min\{n_1,n_2,n_3\}  \rightarrow \infty.
	\end{equation*}

\end{theorem}

For the proof of part (a) of the above theorem, see pp.~151--153 of \cite{korolyuk}. The proof of part (b), i.e., the asymptotic normality of $U_{n_{1}, n_{2}, n_{3}}$, 
follows from \cite{puri1971nonparametric}. 
\\

Since \(\sigma^2\) is unknown, we find its  consistent estimator. Following steps outline how to define a jackknife pseudo-value-based estimator.  Consider\\
(1)  the U statistics based on all observations $U_{n_{1}, n_{2}, n_{3}}$, as defined by \eqref{ustat};\\
(2) for $i=1,2,\dots,n_1,$ the U statistics after deleting observation $X_{i}$, is given by
\begin{small}
	\begin{equation}\label{U-i}
		U_{n_{1}-1, n_{2}, n_{3}}^{(-i)}=
		\frac{1}{{n_1-1\choose m_1}{n_2\choose m_2}{n_3\choose m_3}}  \mathop{\sum\sum\sum}_{\substack{1 \leq i_{1}<\ldots<i_{m_{1}} \leq n_{1}-1 \\ i_{l} \ne i,\;l=1,2,\dots,m_1, \\ 1 \leq j_{1}<\ldots<j_{m_{2}} \leq n_{2} \\ 1 \leq k_{1}<\ldots<k_{m_{3}} \leq n_{3} }} \hskip-0.2in h\left(\bold{X}_{i_{1}}, \ldots, \bold{X}_{i_{m_{1}}}, \bold{Y}_{j_{1}}, \ldots, \bold{Y}_{j_{m_{2}}}, \bold{Z}_{k_{1}} \ldots, \bold{Z}_{k_{m_{3}}}\right);
	\end{equation}
\end{small}\\
(3) for $j=1,2,\dots,n_2,$ the U statistics after deleting observation $Y_{j}$, is given by
\begin{small}
	\begin{equation}\label{U-j}
		U_{n_{1}, n_{2}-1, n_{3}}^{(-j)}=\frac{1}{{n_1\choose m_1}{n_2-1\choose m_2}{n_3\choose m_3}}  \mathop{\sum\sum\sum}_{\substack{1 \leq i_{1}<\ldots<i_{m_{1}} \leq n_{1} \\ 1 \leq j_{1}<\ldots<j_{m_{2}} \leq n_{2}-1 \\ j_{l} \ne j,\;l=1,2,\dots,m_2, \\ 1 \leq k_{1}<\ldots<k_{m_{3}} \leq n_{3} }} \hskip-0.2in  h\left(\bold{X}_{i_{1}}, \ldots, \bold{X}_{i_{m_{1}}}, \bold{Y}_{j_{1}}, \ldots, \bold{Y}_{j_{m_{2}}}, \bold{Z}_{k_{1}} \ldots, \bold{Z}_{k_{m_{3}}}\right);
	\end{equation}
\end{small}\\
(4) for $k=1,2,\dots,n_3,$ the U statistics after deleting observation $Z_{k}$, is given by
\begin{small}
	\begin{equation}\label{U-k}
		U_{n_{1}, n_{2}, n_{3}-1}^{(-k)}=\frac{1}{{n_1\choose m_1}{n_2\choose m_2}{n_3-1\choose m_3}}  \mathop{\sum\sum\sum}_{\substack{1 \leq i_{1}<\ldots<i_{m_{1}} \leq n_{1} \\ 1 \leq j_{1}<\ldots<j_{m_{2}} \leq n_{2} \\ 1 \leq k_{1}<\ldots<k_{m_{3}} \leq n_{3}-1 \\ k_{l} \ne k,\;l=1,2,\dots,m_3 }} \hskip-0.2in  h\left(\bold{X}_{i_{1}}, \ldots, \bold{X}_{i_{m_{1}}}, \bold{Y}_{j_{1}}, \ldots, \bold{Y}_{j_{m_{2}}}, \bold{Z}_{k_{1}} \ldots, \bold{Z}_{k_{m_{3}}}\right).
	\end{equation}
\end{small}
\noindent The jackknife pseudo-values for individual samples are defined as,
\begin{align}\label{V_i00}
	&V_{i, 0,0}=n_{1} U_{n_{1}, n_{2}, n_{3}}-\left(n_{1}-1\right) U_{n_{1}-1, n_{2}, n_{3}}^{(-i)},\; \;\forall \; \; i=1,2,\dots,n_1,\\ \label{V_0j0}
	&V_{0, j, 0}=n_{2} U_{n_{1}, n_{2}, n_{3}}-\left(n_{2}-1\right) U_{n_{1}, n_{2}-1, n_{3}}^{(-j)},\; \;\forall \; \; j=1,2,\dots,n_2
\end{align}
\text{and}
\begin{align}\label{V_00k}
	&V_{0,0, k}=n_{3} U_{n_{1}, n_{2}, n_{3}}-\left(n_{3}-1\right) U_{n_{1}, n_{2}, n_3-1}^{(-k)},\; \;\forall \; \; k=1,2,\dots,n_3.
\end{align}
\vspace{2mm}

\noindent It is easy to verify that $
\bar{V}_{\cdot, 0,0}=\bar{V}_{0, \cdot, 0}=\bar{V}_{0,0, \cdot}=U_{n_{1}, n_{2}, n_{3}}$,
where $$\bar{V}_{\cdot, 0,0}=\frac{1}{n_1}\sum_{i=1}^{n_1}V_{i,0,0},\;\; \bar{V}_{0,\cdot, 0}=\frac{1}{n_2}\sum_{j=1}^{n_2}V_{0,j,0} \; \text{ and } \; \bar{V}_{0,0,\cdot}=\frac{1}{n_3}\sum_{k=1}^{n_3}V_{0,0,k} .$$

The jackknife pseudo-values based estimator of $\sigma^2$, proposed by \cite{arvesen1969jackknifing} (also see \cite{wang2010empirical}, \cite{guangming2013nonparametric} and \cite{an2018jackknife}),is  given as

\begin{align} \label{singma_hat}
	\hat{\sigma}^{2}=: \frac{1}{n_{1}\left(n_{1}-1\right)} \sum_{i=1}^{n_{1}}\left(V_{i, 0,0}-\bar{V}_{\cdot, 0,0}\right)^{2}
	& +\frac{1}{n_{2}\left(n_{2}-1\right)} \sum_{j=1}^{n_{2}}\left(V_{0, j, 0}-\bar{V}_{0,\cdot, 0}\right)^{2} \nonumber\\
	&+\frac{1}{n_{3}\left(n_{3}-1\right)} \sum_{k=1}^{n_{3}}\left(V_{0,0, k}-\bar{V}_{0,0, \cdot}\right)^{2},
\end{align}
where $V_{i,0,0}$, $V_{0,j,0}$ and $V_{0,0,k}$ are defined in equations \eqref{V_i00}, \eqref{V_0j0} and \eqref{V_00k}, respectively.\\

\begin{lemma} \label{lemma of consistency}
\noindent Suppose that $\sigma_{100}^{2}, \sigma_{010}^{2}, \sigma_{001}^{2} > 0$. Then  

$$\hat{\sigma}^{2} - \sigma^{2} = o_{p}\!\left(\min\{n_{1}, n_{2}, n_{3}\}^{-1}\right),$$
where $\sigma^{2} = \frac{m_{1}^{2}\sigma_{100}^{2}}{n_{1}} + \frac{m_{2}^{2}\sigma_{010}^{2}}{n_{2}} + \frac{m_{3}^{2}\sigma_{001}^{2}}{n_{3}},$
and $\hat{\sigma}^{2}$ is defined in \eqref{singma_hat}.
\end{lemma}

Proof: The proof can be easily derived using Chebyshev’s inequality and is therefore omitted (see \cite{guangming2013nonparametric}).
\\

\begin{remark}~

(i) From Theorem \ref{lemma1} we get that $U_{n_1,n_2,n_3}$ is  a consistent estimator of $\theta$.

(ii) From Lemma \ref{lemma of consistency} we get that $\hat{\sigma}^{2}$ is a consistent estimator of $\sigma^{2}$. 

(iii) From Theorem \ref{lemma1} and Lemma \ref{lemma of consistency},it follows that a two-sided $(1-\alpha)$ level asymptotic  CI for $\theta$  is given by 
\begin{equation*}
	\left(U_{n_1,n_2,n_3}-z_{\alpha / 2} \hat{\sigma}, U_{n_1,n_2,n_3}+z_{\alpha / 2} \hat{\sigma}\right),
\end{equation*}

where $z_{\alpha}$ is the upper $\alpha$-th  percentile point of a standard normal distribution.\\

\end{remark}

\vspace*{1mm}
%%%%%%%%%%%%%%%%%%%%%%%%%%%%%%%%%%%%%%%%%%%%%%%%%%%%%%%%%%%%%%%%%%
%%%%%%%%%%%%%%%%%%%%%%%%%%%%%%%%%%%%%%%%%%%%%%%%%%%%%%%%%%%%%%%%%

\subsection{Three sample multivariate  U-Statistics - JEL}

\label{CI on JEL}

 \cite{jing2009jackknife} showed  that the JEL method, based on one and two-sample U-statistics, outperforms the normal approximation method.
We study   generalization of the JEL method to accommodate q-dimensional U-statistics
arising from three  samples.  By broadening the applicability of JEL beyond conventional one and two-sample univariate settings, our approach enhances
its relevance to vector valued three-sample statistical data analysis.
\\
 
 We  combine the three samples and treat them as a single sample. For this purpose,  let $n=n_{1}+n_{2}+n_{3}$ and define $\bold{W}_{l}$ as
$$
\bold{W}_{l}= \begin{cases} \bold{X}_{l}, & \text { if } l=1, \dots, n_{1}  \\ \bold{Y}_{l-n_{1}}, & \text { if } l=n_{1}+1, \dots, n_{1}+n_{2} \\ \bold{Z}_{l-n_{1}-n_{2}}, & \text { if } l=n_{1}+n_{2}+1, \dots, n \end{cases}.
$$

Therefore, $(\bold{W}_1,\bold{W}_2,\dots,\bold{W}_n)=(\bold{X}_1,\dots,\bold{X}_{n_1},\bold{Y}_1,\dots,\bold{Y}_{n_2},\bold{Z}_1,\dots,\bold{Z}_{n_3})$, where $\bold{W}_i's$ are independent but need not be identically distributed. Then, the U-statistic based on single sample $(\bold{W}_1,\bold{W}_2,\dots,\bold{W}_n)$ with degree $m=m_{1}+m_{2}+m_{3}$ is given as follows:
\begin{equation} \label{ustatwhole}
	U_n=\frac{1}{{n\choose m}}\mathop{\sum}_{\substack{1 \leq l_{1}<\ldots<l_{m} \le n}} \tilde{h}\left(\bold{W}_{l_{1}}, \bold{W}_{l_2}, \ldots, \bold{W}_{l_{m}}\right),
\end{equation}
where
\begin{small}
	\\~\\$ \tilde{h}\left(\bold{W}_{l_{1}}, \bold{W}_{l_2}, \ldots, \bold{W}_{l_{m}}\right)$
	$$ = \begin{cases}
		&\text{if} \;\;  1 \leq i_{1}<\ldots<i_{m_{1}} \leq n_{1} \\ c^{*}\;  h\left(\bold{W}_{i_{1}},\dots,\bold{W}_{i_{m_{1}}}, \bold{W}_{j_1}, \ldots,\bold{W}_{j_{m_{2}}},\bold{W}_{k_{1}},\dots, \bold{W}_{k_{m_3}}\right), &  n_1+1 \leq j_{1}<\ldots<j_{m_{2}} \leq n_1+ n_{2} \\ &n_1+ n_{2}+1 \leq k_{1}<\ldots<k_{m_{3}} \leq n\\
		0, &\text{otherwise}
	\end{cases}
	$$
\end{small}
with $c^{*}=\frac{{n\choose m}}{{n_1\choose m_1}{n_2\choose m_2}{n_3\choose m_3}}$. Note that $U_{n_{1}, n_{2}, n_{3}}=U_n$.

\vspace*{2mm}

Consider the U-statistic after deleting observation $\bold{W}_{l}$, from sample $(\bold{W}_1,\bold{W}_2,\dots,\bold{W}_n)$, given by
\begin{equation} \label{ustat_n-1}
	U_{n-1}^{(-l)}=\frac{1}{{n-1\choose m}}\mathop{\sum}_{\substack{1 \leq l_{1}<\ldots<l_{m} \le n\\ l_i\ne l,\; i=1,2,\dots,m}} \tilde{h}\left(\bold{W}_{l_{1}}, \bold{W}_{l_2}, \ldots, \bold{W}_{l_{m}}\right),\;\; \forall \; l=1,2,\dots,n.
\end{equation}

\noindent Define the jackknife pseudo-values $V_l$, for combined  sample $(\bold{W}_1,\bold{W}_2,\dots,\bold{W}_n)$ as
\begin{equation} \label{pseudovalueswhole}
	V_{l}=n U_{n}-(n-1) U_{n-1}^{(-l)},\;\; \forall\; l=1,2,\dots,n.
\end{equation}
It is easy to verify that $ U_n=\frac{1}{n} \sum_{i=1}^{n} V_{i}, $ where $V_{i}$ 's are asymptotically independent random variables (\cite{shi1984approximate}).\\

\noindent The expected values of $V_l's$ are given below (details are provided in Appendix A).
\begin{equation} \label{E[Vi]}
	E V_l= \begin{cases}\theta\left(\frac{n}{n-m}\right)\left[n-m -\frac{(n-1)(n_1-m_1)}{n_1}\right], & \text { if }\; l=1,2, \ldots, n_{1}  \\~\\ \theta\left(\frac{n}{n-m}\right)\left[n-m -\frac{(n-1)(n_2-m_2)}{n_2}\right],  & \text { if } \;l=n_{1}+1, n_{1}+2, \ldots, n_{1}+n_{2} \\~\\
		\theta\left(\frac{n}{n-m}\right)\left[n-m -\frac{(n-1)(n_3-m_3)}{n_3}\right] , & \text { if }\; l=n_{1}+n_{2}+1, \ldots, n \end{cases}.
\end{equation}
It is easy to verify that $\frac{1}{n} \sum_{l=1}^{n} EV_l=\theta$ and, in particular, if $m_{1}=m_{2}=m_{3}$ and $n_{1}=n_{2}=n_{3}$, we get $E V_l=\theta$, $\forall\; l=1,2, \ldots, n$.\\

Finally, we employ the EL method to construct the confidence intervals for $\theta$ by maximizing the empirical likelihood function subject to the constraints based on jackknife pseudo values $V_l'$s. Let $\mathbf{p}=\left(p_{1}, \ldots, p_{n}\right)$ be a probability vector assigning probability $p_{i}$ to each $V_{i}$. The JEL evaluated at $\theta$ becomes
\begin{equation}
	L(\theta)=\max \left\{\prod_{i=1}^{n} p_{i} ; \sum_{i=1}^{n} p_{i}=1, \sum_{i=1}^{n} p_{i}\left(V_{i}-EV_{i}\right)=0\right\}.
\end{equation}

We know that $\prod_{i=1}^{n} p_{i}$ subject to $\sum_{i=1}^{n} p_{i}=1$ attain its maximum at $\frac{1}{n^{n}}$ at $p_{i}=\frac{1}{n}$. Hence the jackknife empirical likelihood ratio for $\theta$ is given by

\begin{equation}
	R(\theta)=\frac{L(\theta)}{n^{-n}}=\max \left\{\prod_{i=1}^{n}\left(n p_{i}\right) ; \sum_{i=1}^{n} p_{i}=1, \sum_{i=1}^{n} p_{i}\left(V_{i}-EV_{i}\right)=0\right\} .
\end{equation}

\noindent Using Lagrange multipliers method, when $$\min _{1 \leq i \leq n}\left(V_{i}-EV_{i}\right)<0<\max _{1 \leq i \leq n}\left(V_{i}-EV_{i}\right),$$ we obtain

\begin{equation} \label{p}
	p_{i}=\frac{1}{n} \frac{1}{1+\lambda\left(V_{i}-EV_{i}\right)}, \;\; \forall \; \; i=1,2,\dots,n,
\end{equation}
where $\lambda$ satisfies

\begin{equation}\label{lambda}
	\frac{1}{n} \sum_{i=1}^{n} \frac{V_{i}-EV_{i}}{1+\lambda\left(V_{i}-EV_{i}\right)}=0.
\end{equation}
Hence we obtain the jackknife empirical log-likelihood ratio as

$$
\log R(\theta)=-\sum_{i=1}^{n} \log \left(1+\lambda\left(V_{i}-EV_{i}\right) \right).
$$
We can use $\log R(\theta)$ for constructing confidence intervals or developing test on $\theta$. For this, we need to obtain the asymptotic distribution of  $\log R(\theta)$. We next derive the limiting distribution of the jackknife empirical log-likelihood ratio as an analog to Wilk's theorem.  The proof of the following theorem is given in Appendix B.

\noindent \begin{theorem} \label{three sample} Suppose that following conditions holds:
	
	\noindent (a) $E\left(h^{2}\left(\bold{X}_{1}, \ldots, \bold{X}_{m_1}, \bold{Y}_{1}, \ldots, \bold{Y}_{m_2}, \bold{Z}_{1} \ldots, \bold{Z}_{m_3}\right)\right)<\infty$ and $\sigma_{100}^{2}, \sigma_{010}^{2}, \sigma_{001}^{2}>0$,
	\vspace*{1.5mm}
	
	\noindent (b) $0<\lim \inf n_{1} / n_{2} \leq \lim \sup n_{1} / n_{2}<\infty$ and  $0<\lim \inf n_{2} / n_{3} \leq \lim \sup n_{2} / n_{3}<\infty$.  \\
	
	\noindent Then, as $n \rightarrow \infty,-2 \log R(\theta)$ converges in distribution to a $\chi^{2}$ random variable with one degree of freedom.
\end{theorem}

\vspace*{1mm}

%%%%%%%%%%%%%%%%%%%%%%%%%%%%%%%%%%%%%%%%%%%%%%%%%%%%%%%%%%
%%%%%%%%%%%%%%%%%%%%%%%%%%%%%%%%%%%%%%%%%%%%%%%%%%%%%%%%%
\section{\textbf{Application: Volumes under ROC surfaces}} \label{CI for VUS}

The receiver operating characteristic (ROC) curve and its associated area under the curve (AUC) are fundamental tools for assessing the diagnostic accuracy of binary classifiers,  summarizing their ability to discriminate between two outcome states across all possible threshold values. In medical science,  ROC analysis is routinely used to evaluate diagnostic tests, compare biomarkers, and determine optimal thresholds for disease detection (\cite{yang2024transformed}). 

 However,  in medical diagnostics, disease severity may exist on a continuum—such as healthy, intermediate, and diseased — rather than fitting neatly into a binary framework. For example, 
consider levels of three biomarkers used to assess the risk of developing cardiovascular disease: let \( X \) be the level of Biomarker 1 (e.g., LDL cholesterol), \( Y \) be the level of Biomarker 2 (e.g., HDL cholesterol), and \( Z \) be the level of Biomarker 3 (e.g., triglycerides). We are interested in finding the probability that a patient has an LDL cholesterol level \( X \) that is lower than their HDL cholesterol level \( Y \), and their HDL cholesterol level \( Y \) that is lower than their triglyceride level \( Z \), i.e., \( P(X < Y < Z) \). This probability, is called Volume under ROC surface and denoted by VUS. It  can be interpreted as a measure of how often the biomarker levels are ordered in a specific way that might be associated with different risk levels or health conditions (\cite{he2008}).\\

This extension  takes care of  the inherent limitations of binary classifications in medical diagnostics, where conditions often manifest along a spectrum rather than as discrete categories. \cite{tian2011exact}. \cite{wang2010empirical}, \cite{guangming2013nonparametric}, and \cite{an2018jackknife} considered non-parametric inference for VUS using the JEL method.\\

Suppose $X_{i}=\left(X_{i 1}, X_{i 2}\right)^{T}, i=1,2, \ldots, n_{1}, Y_{j}=\left(Y_{j 1}, Y_{j 2}\right)^{T}, j=1,2, \ldots, n_{2}$, and $Z_{k}=\left(Z_{k 1}, Z_{k 2}\right)^{T}, k=1,2, \ldots, n_{3},$ represent random samples of three populations having independent joint distribution functions $F, G$ and $H$, respectively. \\

\noindent The difference of two correlated VUS's can be defined as
$$\theta_V =P\left(X_{11}<Y_{11}<Z_{11}\right)-P\left(X_{12}<Y_{12}<Z_{12}\right).$$
In the following part of this subsection, we use three different non-parametric methods for finding  confidence intervals for $\theta_V$. We adopt the notations introduced in Sections \ref{CI for U statistics} and \ref{CI on JEL}, with $q = 2$ and $m_1 = m_2 = m_3 = 1$.

\vspace{6mm}

\noindent \textbf{U-statistic based confidence intervals for $\theta_V$:}\\
\noindent 
 We can rewrite $\theta_V$ as
$$\begin{aligned} \theta_V
	& =E\left[I\left(X_{11}<Y_{11}<Z_{11}\right)-I\left(X_{12}<Y_{12}<Z_{12}\right)\right],
\end{aligned}
$$
where $I\left(x<y<z\right)=\begin{cases}
	1, &\text{if } \; x<y<z\\
	0, &\text{ otherwise}
\end{cases}$.

\noindent Consider the symmetric kernel $h^{VUS}$ given by
$$h^{VUS}\left(\left(X_{i 1}, X_{i 2}\right),\left(Y_{j 1}, Y_{j 2}\right),\left(Z_{k 1}, Z_{k 2}\right)\right)=I\left(X_{i 1}<Y_{j 1}<Z_{k 1}\right)-I\left(X_{i 2}<Y_{j 2}<Z_{k 2}\right). $$
An estimator of $\theta_V$ can be found as U-statistic with kernel $h$ as
\begin{equation*}
	U_{n_1,n_2,n_3}^{VUS}=\frac{1}{n_{1} n_{2} n_{3}} \sum_{i=1}^{n_{1}} \sum_{j=1}^{n_{2}} \sum_{k=1}^{n_{3}} h^{VUS}\left(\left(X_{i 1}, X_{i 2}\right),\left(Y_{j 1}, Y_{j 2}\right),\left(Z_{k 1}, Z_{k 2}\right)\right).
\end{equation*}

The U-statistic $U_{n_1,n_2,n_3}^{VUS}$ is an unbiased and consistent estimator of $\theta$ (Lehmann, (1951)). Also, using Theorem \ref{lemma1}, a two sided asymptotic $(1-\alpha)$ level CI for $\theta_V$ based on the asymptotic distribution of $U_{n_1,n_2,n_3}^{VUS}$ is
\begin{equation} \label{eq:U stat CI}
	\left(U_{n_1,n_2,n_3}^{VUS}-z_{\alpha / 2} \hat{\sigma}, U_{n_1,n_2,n_3}^{VUS}+z_{\alpha / 2} \hat{\sigma}\right),
\end{equation}
where $\hat{\sigma}$ is the jackknife pseudo-values based consistent estimator of $Var(U_{n_1,n_2,n_3}^{VUS})$ (see \cite{arvesen1969jackknifing}), as defined by \eqref{singma_hat}. Note that this confidence interval was also considered in \cite{an2018jackknife}.

\vspace{6mm}

\noindent \textbf{\textbf{JEL based confidence interval for $\theta_V$:}}\\
We use Theorem \ref{three sample} to find JEL based confidence intervals for $\theta_V$. We can define $\log R(\theta_V)$ on the same lines as of Section \ref{CI on JEL}. Hence,  the JEL based confidence interval for $\theta_V$ at level $(1-\alpha)$ is given by
\begin{equation} \label{eq:JEL CI}
    CI =\{ \theta_V\;\vert \;-2 \log R(\theta_V)\leq \chi^2_{1,1-\alpha}\}.
\end{equation}
Note that JEL based confidence interval was also considered in \cite{an2018jackknife}.

\vspace{6mm}

%%%%%%%%%%%%%%%%%%%%%%%%%%%%%%%%%%%%%%%%%%%%%%%%%%%%%%%%%%%%%%%
%%%%%%%%%%%%%%%%%%%%%%%%%%%%%%%%%%%%%%%%%%%%%%%%%%%%%%%%%%%%%%%%%

\noindent \textbf{Kernel based confidence interval for $\theta_V$:}
To estimate the difference of two VUS's \cite{kang2013} proposed the kernel-smoothed version of the difference of two VUS's as follows;
\begin{small}
		\begin{equation} \label{eq:KB}
		{\hat{\theta}_V}=\frac{1}{n_{1} n_{2} n_{3}} \sum_{i=1}^{n_{1}} \sum_{j=1}^{n_{2}} \sum_{k=1}^{n_{3}}\left( \Phi\left(\frac{Y_{1 j}-X_{1 i}}{\sqrt{b_{1}^{2}+b_{2}^{2}}}\right) \Phi\left(\frac{Z_{1 k}-Y_{1 j}}{\sqrt{b_{2}^{2}+b_{3}^{2}}}\right)-  \Phi\left(\frac{Y_{2 j}-X_{2 i}}{\sqrt{b_{4}^{2}+b_{5}^{2}}}\right) \Phi\left(\frac{Z_{2 k}-Y_{2 j}}{\sqrt{b_{5}^{2}+b_{6}^{2}}}\right) \right),
	\end{equation}
\end{small}
\noindent where $\Phi(\cdot)$ represents the standard normal distribution function and the bandwidths are calculated by $$b_{i}=0.9 \min \left\{s d\left(x_{i}\right)\right., \left.\operatorname{iqr}\left(x_{i}\right) / 1.34\right\} n_{i}^{-0.2},\;\;\;\;i=1,2,\dots,6,$$
where $s d(\cdot)$ and $\operatorname{iqr}(\cdot)$ represent the standard deviation and the interquartile range. \cite{silverman2018}  introduced the bandwidth  $b_i$ which controls the level of smoothing in kernel estimation. This asymptotic bandwidth selection method has proven effective across a wide range of density functions and has been widely adopted by researchers in diagnostic studies.\\

We use the estimator (\ref{eq:KB}) and bootstrap method to obtain a kernel-based confidence interval for $\theta_V$. 
In each bootstrap replication, we compute a point estimate $\hat{\theta}_V$ of $\theta_V$, 
which allows us to construct the confidence interval using the percentile method. The kernel-based bootstrap confidence interval for $\theta_V$ is obtained in Step 5 of Algorithm \ref{alg:KB CI}.
This approach is commonly used for constructing confidence intervals; 
see \cite{polansky2001bandwidth}, \cite{feinberg2024kernel}, and \cite{lin2025bootstrap}.

\begin{algorithm}
\caption{Bootstrap Percentile CI for a Kernel-Based Estimator of $\theta_V$}
\label{alg:KB CI}
\begin{algorithmic}[1]
  \State Compute the original estimate $\hat{\theta}_V$ (as defined by \eqref{eq:KB}).
  
  \For{$b = 1,\dots,B$}
    \Statex \textbf{a:} Draw bootstrap samples 
      $(X_1^{*(b)},\dots,X_{n_1}^{*(b)})$, 
      $(Y_1^{*(b)},\dots,Y_{n_2}^{*(b)})$ and 
      $(Z_1^{*(b)},\dots,Z_{n_3}^{*(b)})$ 
      by sampling with replacement from 
      $\{X_1,\dots,X_{n_1}\}$, $\{Y_1,\dots,Y_{n_2}\}$ and 
      $\{Z_1,\dots,Z_{n_3}\}$, respectively.
    \Statex \textbf{b:} Based on these bootstrap samples, compute the bootstrap estimate
      $\hat{\theta}_V^{*(b)}$ using \eqref{eq:KB}.
  \EndFor
  
  \State Sort $\hat{\theta}_V^{*(1)},\dots,\hat{\theta}_V^{*(B)}$ in ascending order to obtain
  $$\hat{\theta}^{*(1)}_{V} \le \cdots \le \hat{\theta}^{*(B)}_{V}.$$
  
  \State Let
  $$\hat{\theta}^*_{\alpha/2} = \hat{\theta}^{*(\lceil B (\alpha/2) \rceil)}_{V}, \quad
    \hat{\theta}^*_{1-\alpha/2} = \hat{\theta}^{*(\lceil B(1-\alpha/2)\rceil)}_{V},$$
  where $\lceil x\rceil$ is the smallest integer greater than or equal to $x$.
  
  \State \textbf{The bootstrap percentile $(1-\alpha)$-confidence interval for $\theta_V$ is
  \begin{equation} \label{eq:KB CI}
      \bigl[\hat{\theta}^*_{\alpha/2},\, \hat{\theta}^*_{1-\alpha/2}\bigr].
  \end{equation}}
\end{algorithmic}
\end{algorithm}

Another possible application of the methodology discussed in Section~\ref{Theory} is the assessment of a diagnostic test based on imaging features. Each patient is represented by a feature vector $\bold{x} \in \Re^q$, and the diagnostic procedure assigns a real-valued score via a function $s : \Re^q \to \Re$. Patients are to be classified into three ordered categories (e.g., healthy, mildly diseased, severely diseased) on the basis of this score.

Let $\bold{X}$, $\bold{Y}$, and $\bold{Z}$ denote independent observations from the healthy, mildly diseased, and severely diseased populations, respectively. The parameter of interest is
$$
  \theta_s = P\big(s(\bold{X}) < s(\bold{Y}) < s(\bold{Z})\big),$$
which is the probability that a randomly selected triplet, one patient from each group, is correctly ordered by the scoring function $s(\cdot)$. A natural point estimator of $\theta_s$ provides a summary measure of the accuracy of the diagnostic test, and one may also be interested in constructing a confidence interval for $\theta_s$.

\section{\textbf{Simulation}}

In this section, we employ Monte Carlo simulation studies to evaluate the finite sample performance of CIs for difference of two VUS's. The three methods used are the JEL, normal approximation with jackknife variance estimator and kernel based bootstrapping. We assess the performance of each method using two distinct criteria:

\noindent(a) Coverage Probability (CP): This represents the probability that the true parameter value falls within the constructed confidence interval (CI). Methods with lower coverage error—defined as the discrepancy between the actual coverage probability and the nominal value are preferred.\

\noindent(b) Average Length (AL) of CIs: Methods that produce CIs with shorter average lengths are desirable, as shorter intervals provide more precise information about the location of the unknown parameter.

We conduct simulation studies with different sample sizes and simulation is repeated  1,000 time.  We use R package ``empilik" for finding  JEL based confidence intervals. For constructing confidence intervals based on kernel based method we use 1000 bootstrap replications. Two bivariate probability models are considered for the simulation studies.
\vspace*{2mm}

\noindent \textbf{Marshall-Olkin bivariate exponential model (MOBVE($\lambda_{1},\lambda_{2},\lambda_{3}$)) :}
We consider the Marshall-Olkin bivariate exponential model MOBVE($\lambda_1, \lambda_2, \lambda_3$), introduced by \cite{marshall1967multivariate}. The joint cumulative distribution function of this bivariate model is given by
$$F_{X_1,X_2}(x_1,x_2)=1-\exp[-\lambda_1 x_1-\lambda_2 x_2 -\lambda_3 \max\{x_1,x_2\}],\;\; x_1>0,\; x_2>0,$$
where $\lambda_i>0,\;i=1,2,3.$
We generated sample sizes $n_1$, $n_2$, and $n_3$ from three  independent models MOBVE($\lambda_{x_1},\lambda_{x_2},\lambda_{x_3}$), MOBVE($\lambda_{y_1},\lambda_{y_2},\lambda_{y_3}$), and MOBVE($\lambda_{z_1},\lambda_{z_2},\lambda_{z_3}$), respectively.
\vspace*{2mm}

In Table \ref{table1}, we explored various configurations for sample sizes $(n_1,n_2,n_3)$ and the parameters $(\lambda_{x_1},\lambda_{x_2},\lambda_{x_3};\lambda_{y_1},\lambda_{y_2},\lambda_{y_3};\lambda_{z_1},\lambda_{z_2},\lambda_{z_3})$. This
comprehensive investigation aims to evaluate the performance and efficacy of these three methods
under various parameter settings within the MOBVE framework.
\\

\vspace*{1mm}

\noindent \textbf{Bivariate Pareto model :}
In this part, we explore the Farlie-Gumbel-Morgenstern (FGM) copula based on two independent Pareto distributions with C.D.F.s defined as follows:
$$F_{X_i}(x_i)= 1- \left(\frac{\lambda_i}{x_i}\right)^{\alpha_i},\; x_i\geq \lambda_i >0,\; \alpha_i>0,\;i=1,2.$$

\noindent For $-1 \leq \theta \leq 1$, the FGM copula, or joint C.D.F., is given by
\begin{eqnarray*}
	% \nonumber % Remove numbering (before each equation)
	C_{\theta}(F_{X_1}(x_1),F_{X_2}(x_2))&=& F_{X_1}(x_1)F_{X_2}(x_2)[1+\theta (1-F_{X_1}(x_1))(1-F_{X_2}(x_2))],\\&&\qquad\qquad\qquad
	0\leq F_{X_1}(x_1), \, F_{X_2}(x_2) \leq 1.
\end{eqnarray*}

We refer to this bivariate model as the FGM-based Bivariate Pareto model (FGMPM), denoted as FGMPM($\lambda_1, \lambda_2, \alpha_1,\alpha_2,\theta$). For simulations, we generated sample sizes $n_1$, $n_2$, and $n_3$ from three FGMPM models:
$ FGMPM(\lambda_{x_1}, \lambda_{x_2}, \alpha_{x_1},\alpha_{x_2},\theta_x),\; FGMPM(\lambda_{y_1}, \lambda_{y_2}, \alpha_{y_1},\alpha_{y_2},\theta_y),$
and  $FGMPM(\lambda_{z_1}, \lambda_{z_2}, \alpha_{z_1},\alpha_{z_2},\theta_z),$
respectively. The samples were generated using the ``rCopula" function in the ``copula" package in R. Results of the simulation studies for this model are presented in Table \ref{table2} for various configurations of sample sizes $(n_1, n_2, n_3)$ and parameters $(\lambda_{x_1}, \lambda_{x_2}, \alpha_{x_1}, \alpha_{x_2}, \theta_x; \lambda_{y_1}, \lambda_{y_2}, \alpha_{y_1}, \alpha_{y_2}, \theta_y; \lambda_{z_1}, \lambda_{z_2}, \alpha_{z_1}, \alpha_{z_2}, \theta_z)$.

\vspace{2mm}

\noindent The simulations were conducted on a laptop with the following hardware configuration: Processor: Intel(R) Core(TM) i7-7500U CPU @ 2.70GHz, 4 cores; Memory: 8 GB DDR4 RAM; Storage: 240 GB SSD WD green; Operating System: Windows 10, 64-bit.

\noindent  The approximate time taken by the three methods for obtaining results during simulations is presented in Table \ref{time1}.

The following observations can be made from Tables \ref{time1}- \ref{table2}:\\

1. We observe that for various configurations of parameters of the Marshall-Olkin bivariate exponential and bivariate Pareto models, the average lengths (ALs) of the confidence intervals based on JEL, normal approximation, and kernel-based methods decrease as the sample sizes increase.\\

2. In most cases, as the sample size increases, the coverage probabilities of the confidence intervals generated by the three methods tend to align with the nominal level (95\%). This is anticipated because larger sample sizes yield more detailed data insights.\\

3. From Tables \ref{table1}-\ref{table2}, we can see that CI based on JEL method achieves closer to 95\% coverage across various choices of parameters compared to the other two methods. The average interval lengths are generally better for CIs based on the normal approximation and kernel-based methods for small sample sizes. However, the average lengths are almost same for moderate sample sizes.\\

4. From Table \ref{time1}, we can observe that the simulation time taken using the normal approximation with the jackknife variance estimator and kernel-based bootstrapping methods is significantly higher than the time taken by the JEL method.\\

Therefore, based on the simulation studies, we recommend the JEL method as it requires significantly less time to construct confidence intervals compared to the other two methods and also provides better coverage probability.
\\

Note that \cite{an2018jackknife} also conducted simulation studies for the confidence intervals defined in \eqref{eq:U stat CI} and \eqref{eq:JEL CI}. However, they did not consider the bivariate Pareto model and did not compare these methods with kernel-based confidence intervals.

%	\FloatBarrier
\begin{table}[h]
	\caption{Simulation time taken by various methods under two models}
	\begin{tabular}{l l l l l}
		\hline
		Model &	\begin{tabular}[c]{@{}l@{}}Samples Size\\ $(n_1,n_2,n_3)$\end{tabular} & JEL         & Normal Approx & Kernel     \\
		\hline
		\multirow{2}{*}{MOBVE} &	$(10,10,10)$                                                            & 1.6 minutes & 1.95 minutes  & 24 minutes \\
		
		&	$(100,100,100)$                                                         & 40 hours    & 69 hours      & 240 hours \\
		\multirow{2}{*}{bivariate Pareto} &	$(10,10,10)$                                                            & 2.63 minutes & 4.6 minutes  & 50 minutes \\
		&	$(100,100,100)$                                                         & 135 hours    & 246 hours      & 656 hours \\
		\hline
	\end{tabular}
	\label{time1}
\end{table}
%	\FloatBarrier

%\restoregeometry

\FloatBarrier
%\newgeometry{right=1.5cm, bottom=1.5cm}

\begin{table}[h]
	\footnotesize
	
	\caption{Coverage probability (CP) and average length (AL) of confidence intervals for the difference two VUS's at 95\% confidence level: Under MOBVE models}
	\centering
	\begin{adjustbox}{angle=90}
		\begin{tabular}{ll  |ll | ll | ll }
			\hline
			\multicolumn{1}{c}{\multirow{2}{*}{$(\lambda_{x_1},\lambda_{x_2},\lambda_{x_3};\lambda_{y_1},\lambda_{y_2},\lambda_{y_3};\lambda_{z_1},\lambda_{z_2},\lambda_{z_3})$}} & \multicolumn{1}{c}{\multirow{2}{*}{$(n_1,n_2,n_3)\;\;\;$}} & \multicolumn{2}{c}{JEL}  & \multicolumn{2}{c}{Normal Approx} & \multicolumn{2}{c}{Kernel}  \\ \cline{3-4} \cline{5-6} \cline{6-8}
			\multicolumn{1}{c}{}                                     & \multicolumn{1}{c}{}                            & \multicolumn{1}{c}{CP (\%)} & \multicolumn{1}{c}{AL} & \multicolumn{1}{c}{CP (\%)} & \multicolumn{1}{c}{AL} & \multicolumn{1}{c}{CP (\%)} & \multicolumn{1}{c}{AL} \\ \hline
			
			(1,1,1;1,1,1;1,1,1)	& (10,10,10)                                                &      96.8                      &    .501   & 97.8 & 0.384         & 97 & 0.392      \\
			&    (20,20,20)                                           &     96.5                         &      .332     & 96.8 & 0.325 &95.8 & 0.341             \\
			&  (30,30,30)                                               &  95.7                          &    .229        &97 & 0.257 & 96.6 & 0.264            \\
			&     (50,50,50)                                                &   95                      &      .18     & 96.6 & 0.183 & 96 & 0.195              \\
			
			&      (70,70,70)                                           &   94.4                          &    .126    & 94 & 0.124 & 93.5 & 0.129                \\
			&  (100,100,100)                                               &   95.8                          &       .102      & 94.2 & 0.103 & 94 & 0.122            \\

			&                                                 &                             &                        \\
			(1,2,0;1,1,0;2,1,0)	& (10,10,10)                                                &    97                         &       .564   & 92.8 & 0.516 & 91.8 & 0.504              \\
			&      (20,20,20)                                           &    96.4                        &     .35          & 93 & 0.361 & 92 & 0.341         \\
			&  (30,30,30)                                               &     95.8                        &     .248        & 92.5 & 0.239 & 93 & 0.25        \\
			&     (50,50,50)                                                &  95.4                       &   .187     & 94 & 0.187 & 92.5 & 0.169                \\
			&      (70,70,70)                                           &    95.2                         &   .157    & 94 & 0.187 & 92.5 & 0.158                 \\
			&  (100,100,100)                                               &         95.2                    &     .131     & 94.5 & 0.130 & 94 & 0.135              \\~\\
			(1/3,1/3,2/3;2/3,2/3,4/3;1,1,2)	& (10,10,10)                                                &    99.4                         &       .263      & 99.6 & 0.234 & 97.5 & 0.21           \\
			&      (20,20,20)                                           &    97                       &     .169  & 97.5 & 0.151 & 98.5 & 0.161                 \\
			&  (30,30,30)                                               &     96                       &     .112    & 97.5 & 0.121 & 98.25 & 0.125               \\
			&     (50,50,50)                                                &  95.6                     &   .085    & 95.75 & 0.096 & 97.5 & 0.089                 \\
			&      (70,70,70)                                           &    95.6                         &   .067       & 94.5 & 0.069 & 96.5 & 0.071              \\
			&  (100,100,100)                                               &      95.2                       &
			.053 & 94.8 & 0.058 & 96 & 0.062  \\
			\\~\\
			(3/5,3/5,2/5;6/5,6/5,4/5;9/5,9/5,6/5)	& (10,10,10)                                                &    99.2                         &       .319           & 98.75 & 0.269 & 96.5 & 0.235      \\
			&      (20,20,20)                                           &    95.2                       &     .2  & 97 & 0.197 & 97.25 & 0.211                 \\
			&  (30,30,30)                                               &     97                       &     .146  & 96.5 & 0.151 & 97 & 0.142                 \\
			&     (50,50,50)                                                &  96.4                     &   .103   & 96.5 & 0.098 & 97 & 0.106                  \\
			&      (70,70,70)                                           &    96                         &   .087     & 96 & 0.082 & 96.5 & 0.09                \\
			&  (100,100,100)                                               &      95.6                     &
			.066   & 96 & 0.069 & 96 & 0.065
			\\~\\
			(1/19,1/19,18/19;2/19,2/19,36/19;3/19,3/19,54/19)	& (10,10,10)                                                &                          98   &          0.09 & 86 & 0.07 & 88 & 0.079            \\
			&      (20,20,20)                                           &    95.2                       &       0.071   & 98 & 0.06 & 98 & 0.055             \\
			&  (30,30,30)                                               &     97                       &     0.048             & 99 & 0.052 & 94 & 0.049      \\
			&     (50,50,50)                                                &  91                     &   0.034   & 89 & 0.034 & 88.5 & 0.036                  \\
			&      (70,70,70)                                           &    90.4                         &   0.027         & 90 & 0.028 & 88 & 0.029             \\
			&  (100,100,100)                                               &      93.6                     &
			0.023 & 91 & 0.022 & 90  & 0.025  \\
			\hline
		\end{tabular}
	\end{adjustbox}
	\label{table1}
\end{table}
\FloatBarrier

%\restoregeometry

\FloatBarrier
%\newgeometry{left=1.5cm, bottom=1.5cm}
\begin{table}[h]
	\footnotesize
	\caption{Coverage probability (CP) and average length (AL) of confidence intervals for the difference two VUS's at 95\% confidence level: Under Bivariate Pareto models}
	\centering
	\begin{adjustbox}{angle=90}
		\begin{tabular}{ll  |ll | ll | ll }
			\hline
			\multicolumn{1}{c}{\multirow{2}{*}{$(\lambda_{x_1},\lambda_{x_2},\alpha_{x_1},\alpha_{x_2},\theta_x;\lambda_{y_1}, \lambda_{y_2}, \alpha_{y_1},\alpha_{y_2},\theta_y;\lambda_{z_1}, \lambda_{z_2}, \alpha_{z_1},\alpha_{z_2},\theta_z)$}} & \multicolumn{1}{c}{\multirow{2}{*}{$(n_1,n_2,n_3)\;\;\;$}} & \multicolumn{2}{c}{JEL}  & \multicolumn{2}{c}{Normal Approx} & \multicolumn{2}{c}{Kernel}  \\ \cline{3-4} \cline{5-6} \cline{6-8}
			\multicolumn{1}{c}{}                                     & \multicolumn{1}{c}{}                            & \multicolumn{1}{c}{CP (\%)} & \multicolumn{1}{c}{AL} & \multicolumn{1}{c}{CP (\%)} & \multicolumn{1}{c}{AL} & \multicolumn{1}{c}{CP (\%)} & \multicolumn{1}{c}{AL} \\ \hline
			
			(1,1,1,1,0.5;1,1,1,1,0.5;1,1,1,1,-0.5)	& (10,10,10)                                                &     96                        &  .493 & 96.3 & 0.451 & 98 & 0.485                  \\
			&      (20,20,20)                                           &  97                           &  .311            & 97 & 0.298 & 96.6 & 0.322          \\
			&  (30,30,30)                                               &    95                         &    .274     & 94.3 & 0.262 & 94.1 & 0.281               \\
			&     (50,50,50)                                                &     96.8                    &   .192     & 93.8 & 0.201 & 93.7 & 0.182                \\
			
			&      (70,70,70)                                           &     95.2                        &   .156   & 94.8 & 0.153 & 94 & 0.159                  \\
			&  (100,100,100)                                               &     95.4                        &   .109       & 94.8 & 0.110 & 94.2 & 0.101              \\

			&                                                 &                             &                        \\~\\
			(1,2,1,2,0.5;2,1,2,1,0.5;2,2,1,1,-0.5)	& (10,10,10)                                                &   96                          &     .557   & 91 & 0.549 & 89 & 0.514               \\
			&      (20,20,20)                                           &                          98   &    .362  & 99 & 0.38 & 91 & 0.391                  \\
			&  (30,30,30)                                               &                          92   &     .268     & 92 & 0.306 & 91 & 0.321              \\
			&     (50,50,50)                                                &    98                     &     .193       & 99 & 0.235 &93 & 0.242            \\
			&      (70,70,70)                                           &                          97   &     .198    & 94 & 0.189 & 92.5 & 0.191               \\
			&  (100,100,100)                                               &      96                       &    .183      & 94 & 0.185 & 94 & 0.187              \\~\\
			(5,1,1,5,0.2;5,1,2,2,0.5;1,5,5,1,0.9)	& (10,10,10)                                                &   97                          &    .460  & 90 & 0.465 & 88 & 0.45                  \\
			&      (20,20,20)                                           &                          93   &       .345             & 92 & 0.325 & 91 & 0.34    \\
			&  (30,30,30)                                               &                           94  &     .275     & 92 & 0.295 & 92 & 0.29              \\
			&     (50,50,50)                                                &                        96 &    .174     & 97 & 0.17 & 98 & 0.183               \\
			&      (70,70,70)                                           &     96                        &    .165    & 93.5 & 0.166 & 94 & 0.169                \\
			&  (100,100,100)                                               &   95.7                          &   .151       & 94 & 0.150 & 94 & 0.153              \\~\\
			(0.5,1,0.5,1,-0.5;1,1,0.5,0.5,-0.5;1,1,5,1,-0.2)	& (10,10,10)                                                &    94                          &   .359         & 92 & 0.356 & 91 & 0.341            \\
			&      (20,20,20)                                           &   98                          &   .212     & 93 & 0.201 & 92 & 0.204                 \\
			&  (30,30,30)                                               &                      97       &     .169   & 93 & 0.157 &     92 & 0.17             \\
			&     (50,50,50)                                                &   96                      &   .127          & 93 & 0.13 & 92.5 & 0.129           \\
			&      (70,70,70)                                           &    95.8                         &   .098        & 93.5 & 0.095 & 93 & 0.097             \\
			&  (100,100,100)                                               &      95.5                       &     .076       & 94 & 0.075 & 93.5 & 0.076            \\~\\
			(10,1,0.5,10,-0.9;1,1,15,0.2,-0.1;0.2,1,5,0.1,0.9)	& (10,10,10)                                                &                        92     &       .530      & 90 & 0.521 & 89 & 0.532           \\
			&      (20,20,20)                                           &                           98  &    .317       & 91 & 0.29 & 89 & 0.299             \\
			&  (30,30,30)                                               &                     97        &      .262      & 91 & 0.267 & 91 & 0.276            \\
			&     (50,50,50)                                                &                     96    &               .196   & 92 & 0.196 & 93 & 0.189      \\
			&      (70,70,70)                                           &                       96      &                .163    & 98 & 0.165 & 99 & 0.17    \\
			&  (100,100,100)                                               &                       95.8      &                   .123  & 94 & 0.12 & 93.5 & 0.125   \\
		\end{tabular}
	\end{adjustbox}
	\label{table2}
\end{table}
\FloatBarrier

%\restoregeometry

\section{\textbf{Data analysis}}
In this section, we apply different methods discussed above to the Alzheimer’s disease dataset taken from the Alzheimer’s Disease Neuroimaging Initiative (ADNI\footnote{ADNI, \url{https://adni.loni.usc.edu}} (2024)). In Alzheimer’s disease research, biomarkers like amyloid-beta (A$\beta$) and phosphorylated tau (p-tau) are critical for understanding disease progression. The focus is on understanding how these biomarkers behave differently in patients with Alzheimer’s disease (AD), cognitively normal (CN) individuals, and those with late mild cognitive impairment (LMCI). This analysis aims to compare the differences in volumes under surfaces that represent probabilities involving the amyloid-beta (A$\beta$) and phosphorylated tau (p-tau) biomarkers across various cognitive states in Alzheimer’s disease.
Let random variables $X_1$, $Y_1$, and $Z_1$ represent A$\beta$ levels in AD, CN, and LMCI, respectively, and random variables $X_2$, $Y_2$, and $Z_2$ represent p-tau levels in AD, CN, and LMCI, respectively. Here, our parameter of interest is
$$\theta_V=P(X_1 < Y_1 < Z_1)-P(X_2 < Y_2 < Z_2),$$
where $P(X_1 < Y_1 < Z_1)$ is the probability that an individual with AD has lower A$\beta$ levels compared to a CN individual, and the CN individual has lower A$\beta$ levels compared to an individual with LMCI. If this probability is high, it indicates that the A$\beta$ biomarker effectively distinguishes between the three cognitive states. The probability $P(X_2 < Y_2 < Z_2)$ can be interpreted similarly.\\

This comparison tells us which biomarker, A$\beta$ or p-tau, is more effective in distinguishing among the three cognitive conditions. A positive estimate of $\theta_V$ suggests that A$\beta$ levels provide a clearer characterization of the progression from normal cognition to late mild impairment to Alzheimer's disease, whereas a negative value of $\theta_V$ indicates greater discriminatory power for p-tau levels.
\\

In total, 883 results were included in this analysis, with sample sizes for the three individual classes being 222, 122, and 539, respectively. We used the JEL, Normal approximation, and Kernel-based methods, as described in Section \ref{CI for VUS}, to construct confidence intervals for the difference between two VUSs. The JEL method took 3.4 hours, the Normal approximation method took 4.7 hours, and the Kernel-based method (with 1000 bootstraps) took 52 hours to obtain results for this real data analysis using R software on the previously specified laptop. \\

In Table \ref{real data table}, we presented the confidence intervals calculated using the three methods discussed above. Since the lower bounds of confidence intervals are positive, the A$\beta$ biomarker is more effective in highlighting the changes or differences among the three cognitive conditions. From the simulation studies, we found that the JEL method performs better, and therefore we recommend JEL-based confidence intervals for the difference of two VUSs.
\\

Note that, for the application, we consider a recent  dataset where the biomarkers and disease stages are entirely distinct from those used in the real-data analysis in \cite{an2018jackknife}.

%\FloatBarrier
\begin{table}[h]
	\caption{Confidence Intervals for $\theta$ at different confidence levels using different methods}
	\label{real data table}
	\begin{tabular}{l l l l }
		\hline
		Method              & JEL             & Normal Approx   & Kernel     \\ \hline \hline
		Confidence Interval 90\% & (0.3124,0.3753) & (0.3161,0.3723) & (0.3083,0.3743)\\
		
		Confidence Interval 95\% & (0.3062,0.3813) & (0.3107,0.3776) & (0.3006,0.3803) \\
		Confidence Interval 99\% & (0.3007,0.3866) & (0.3002,0.3882) & (0.2909,0.3904)\\ \hline
	\end{tabular}
\end{table}
%\FloatBarrier

%%%%%%%%%%%%%%%%%%%%%%%%%%%%%%%%%%%%%%%%%%%%%%%%%%%%%%%%%%

\section{\textbf{Concluding remarks}}
%  extended the jackknife empirical likelihood (JEL) method, originally proposed by Jing et al. (2009) for one and two univariate sample U-statistics,
In this paper, we studied U-statistics based on three multivariate samples and developed a jackknife empirical likelihood (JEL) methodology for such statistics. This would extend the  applicability of JEL methods to three sample 
multivariate problems.\\

We constructed confidence intervals for differences in two dependent  VUSs using three nonparametric procedures. We conducted an  extensive simulation study under the Marshall-Olkin bivariate exponential and Bivariate Pareto distributions. Our results showed the superior performance of the JEL method in terms of coverage probability and computational complexity when compared to the normal approximation using jackknife pseudo-values based variance estimator and kernel-based methods. \\

Additionally, we conducted a real data analysis using a recent Alzheimer's disease dataset to illustrate the practical benefits of the JEL method in real-world scenarios. This study highlights the superior performance, versatility, lower computational cost and practical benefits of the JEL method over its competitors recommending its use in statistical practices that require reliable nonparametric inference for complex data structures.\\

%Researchers,  including \cite{subhash1986}, \cite{mahajan2011}, and \cite{deshpande2017} (chapter 8)  have explored  multisample testing problems based on U-statistics.The theory of U-statistics and jackknife empirical likelihood can be used to test nonparametric hypotheses and construct distribution free  asymptotic confidence intervals for three sample data based on vectors. \\

Beyond diagnostic testing, the proposed framework is also relevant in other scientific domains. For example, an ecologist could be interested in comparison of multiple species in three habitats: protected forest, logged forest and agricultural land. He would be interested in testing the hypothesis that the distribution of multiple species in these three habitats are identical. \cite{rizzo2016energy} have discussed  energy based statistics for this testing problem. Many of them can be expressed as U-statistics. Asymptotic distributions of the test statistics (either U-statistics or JEL ratio) will follow from the results in Section \ref{Theory}. 
\\

We are working on 
$q$-dimensional $k$-sample U-statistics and the asymptotic distribution of the associated JEL ratio. One possible application of the theory developed would be construction of  confidence intervals for hypervolume under the ROC manifold  which is a generalization of ROC curves 
to  multi-category classification.  The results will be presented elsewhere.\\

%%%%%%%%%%%%%%%%%%%%%%%%%%%%%%%%%%%%%%%%%%%%%%%%%%%%%%%%%%%%%%%%%
%%%%%%%%%%%%%%%%%%%%%%%%%%%%%%%%%%%%%%%%%%%%%%%%%%%%%%%%%%%%%%%%%%

\allowdisplaybreaks

\bibliographystyle{apalike}
\bibliography{references}

%% or include bibliography directly:

\section*{\textbf{Appendix A}: \textbf{For obtaining} $EV_l$}\label{ApA}

\noindent We simplify the expression for $V_l's$ and then find $EV_l$. \\

\noindent For simplicity, denote
\begin{small}
	$$\bar{h}_{i_1,\dots,i_{m_1},j_1,\dots,j_{m_2},k_1,\dots,k_{m_3}}\hskip-0.1in=h(\bold{X}_{i_{1}}, \ldots, \bold{X}_{i_{m_{1}}}, \bold{Y}_{j_{1}-n_1}, \ldots, \bold{Y}_{j_{m_{2}}-n_1}, \bold{Z}_{k_{1}-n_1-n_2} \ldots, \bold{Z}_{k_{m_{3}}-n_1-n_2}).$$
\end{small}
For $l=1,2,\dots,n$, we have
\begin{align} \label{V_l expression}
	V_l=& n U_{n}-(n-1) U_{n-1}^{(-l)} \nonumber\\
	=& n U_{n}-\frac{(n-1)}{{n-1 \choose m}} \mathop{\sum}_{\substack{1 \leq l_{1}<\ldots<l_{m} \le n \nonumber\\ l_i\ne l,\; i=1,2,\dots,m}} \tilde{h}\left(\bold{W}_{l_{1}}, \bold{W}_{l_2}, \ldots, \bold{W}_{l_{m}}\right) \nonumber\\
	= & n U_{n}
	-\frac{(n-1)}{{n-1 \choose m}} \mathop{\sum\sum\sum}_{\substack{1 \leq i_{1}<\ldots<i_{m_{1}} \\ <  j_{1}<\ldots<j_{m_{2}}  \\ < k_{1}<\ldots<k_{m_{3}} \leq n-1 \\ i_i\le n_1< j_j \le n_1+n_2< k_k, \\	\forall \; i=1,\dots,m_1,\; j=1,\dots,m_2,\; k=1,\dots,m_3\\ i_i=j_i= k_i \ne l,\;i=1,2,\dots,m }} \frac{{n\choose m}}{{n_1\choose m_1}{n_2\choose m_2}{n_3\choose m_3}}  \bar{h}_{i_1,\dots,i_{m_1},j_1,\dots,j_{m_2},k_1,\dots,k_{m_3}} \nonumber\\
	= & n U_{n}-\frac{n(n-1)}{n-m} \frac{1}{{n_1\choose m_1}{n_2\choose m_2}{n_3\choose m_3}} \Bigg[ \mathop{\sum\sum\sum}_{\substack{1 \leq i_{1}<\ldots<i_{m_{1}} \leq n_{1}-1 \\ i_{i} \ne l,\;i=1,2,\dots,m_1 \\ n_1+1 \leq j_{1}<\ldots<j_{m_{2}} \leq n_1+ n_{2} \\ n_1+n_2+1 \leq k_{1}<\ldots<k_{m_{3}} \leq n }}  \hskip-0.2in I_{[1,n_1]}(l)\, \bar{h}_{i_1,\dots,i_{m_1},j_1,\dots,j_{m_2},k_1,\dots,k_{m_3}} \nonumber\\
	& \qquad \qquad \qquad  + \mathop{\sum\sum\sum}_{\substack{1 \leq i_{1}<\ldots<i_{m_{1}} \leq n_{1} \\  n_1+1 \leq j_{1}<\ldots<j_{m_{2}} \leq n_1+ n_{2} -1 \\ j_{j} \ne l,\;j=1,2,\dots,m_2 \\ n_1+n_2+1 \leq k_{1}<\ldots<k_{m_{3}} \leq n }}   I_{[n_1+1,n_1+n_2]}(l)\, \bar{h}_{i_1,\dots,i_{m_1},j_1,\dots,j_{m_2},k_1,\dots,k_{m_3}} \nonumber\\
	& \qquad \qquad \qquad  +\mathop{\sum\sum\sum}_{\substack{1 \leq i_{1}<\ldots<i_{m_{1}} \leq n_{1}  \\ n_1+1 \leq j_{1}<\ldots<j_{m_{2}} \leq n_1+ n_{2} \\ n_1+n_2+1 \leq k_{1}<\ldots<k_{m_{3}} \leq n-1 \\ k_{k} \ne l,\;k=1,2,\dots,m_3 }}   I_{[n_1+n_2+1,n]}(l)\, \bar{h}_{i_1,\dots,i_{m_1},j_1,\dots,j_{m_2},k_1,\dots,k_{m_3}}  \Bigg],
\end{align}
\allowdisplaybreaks
where $I_{[a,b]}(x)$ is a indicator function, defined as $I_{[a,b]}(x)=\begin{cases}
	1, &\text{if } a\leq x\leq b\\
	0, & \text{otherwise}
\end{cases}.$

\vspace{2mm}

Since $U_{n_{1}, n_{2}, n_{3}}=U_n$ and combining \eqref{U-i} and \eqref{V_i00}, \eqref{U-j} and \eqref{V_0j0}, and \eqref{U-k} and \eqref{V_00k}, we have the following expressions,
\begin{align} \label{sum-i}
	\text{(1) for $i=1,2,\dots,n_1,$} \qquad \qquad \qquad & \nonumber\\
	\mathop{\sum\sum\sum}_{\substack{1 \leq i_{1}<\ldots<i_{m_{1}} \leq n_{1}-1 \\ i_{l} \ne i,\;l=1,2,\dots,m_1 \\ 1 \leq j_{1}<\ldots<j_{m_{2}} \leq n_{2} \\ 1 \leq k_{1}<\ldots<k_{m_{3}} \leq n_{3} }} h(\bold{X}_{i_{1}}, \ldots, \bold{X}_{i_{m_{1}}},&\bold{Y}_{j_{1}}, \ldots, \bold{Y}_{j_{m_{2}}}, \bold{Z}_{k_{1}} \ldots, \bold{Z}_{k_{m_{3}}}) \nonumber\\
	&={n_1-1\choose m_1}{n_2\choose m_2}{n_3\choose m_3} \left( \frac{n_1 U_n}{n_1-1}-\frac{V_{i,0,0}}{n_1-1} \right);
\end{align}
\begin{align} \label{sum-j}
	\text{(2) for $j=1,2,\dots,n_2,$}  \qquad \qquad \qquad& \nonumber\\
	\mathop{\sum\sum\sum}_{\substack{1 \leq i_{1}<\ldots<i_{m_{1}} \leq n_{1} \\ 1 \leq j_{1}<\ldots<j_{m_{2}} \leq n_{2}-1 \\ j_{l} \ne j,\;l=1,2,\dots,m_2 \\ 1 \leq k_{1}<\ldots<k_{m_{3}} \leq n_{3} }} h(\bold{X}_{i_{1}}, \ldots, \bold{X}_{i_{m_{1}}},& \bold{Y}_{j_{1}}, \ldots, \bold{Y}_{j_{m_{2}}}, \bold{Z}_{k_{1}} \ldots, \bold{Z}_{k_{m_{3}}})\nonumber\\
	&={n_1\choose m_1}{n_2-1\choose m_2}{n_3\choose m_3} \left( \frac{n_2 U_n}{n_2-1}-\frac{V_{0,j,0}}{n_2-1} \right);
\end{align}
\begin{align} \label{sum-k}
	\text{(3) and for $k=1,2,\dots,n_3,$}  \qquad \qquad& \nonumber\\
	\mathop{\sum\sum\sum}_{\substack{1 \leq i_{1}<\ldots<i_{m_{1}} \leq n_{1} \\ 1 \leq j_{1}<\ldots<j_{m_{2}} \leq n_{2} \\ 1 \leq k_{1}<\ldots<k_{m_{3}} \leq n_{3}-1\\ k_{l} \ne k,\;l=1,2,\dots,m_3 }} h(\bold{X}_{i_{1}}, \ldots, \bold{X}_{i_{m_{1}}}, &\bold{Y}_{j_{1}}, \ldots, \bold{Y}_{j_{m_{2}}}, \bold{Z}_{k_{1}} \ldots, \bold{Z}_{k_{m_{3}}})\nonumber\\
	&= {n_1\choose m_1}{n_2\choose m_2}{n_3-1\choose m_3} \left( \frac{n_3 U_n}{n_3-1}-\frac{V_{0,0,k}}{n_3-1} \right).
\end{align}
Substituting \eqref{sum-i}, \eqref{sum-j} and \eqref{sum-k} into \eqref{V_l expression}, and we obtain
\begin{small}
	\begin{align}
		V_l= & n U_{n}-\frac{n(n-1)}{n-m} \frac{1}{{n_1\choose m_1}{n_2\choose m_2}{n_3\choose m_3}} \Bigg[	{n_1-1\choose m_1}{n_2\choose m_2}{n_3\choose m_3} \left( \frac{n_1 U_n}{n_1-1}-\frac{V_{l,0,0}}{n_1-1} \right)I_{[1,n_1]}(l)\nonumber \\
		& \qquad + {n_1\choose m_1}{n_2-1\choose m_2}{n_3\choose m_3} \left( \frac{n_2 U_n}{n_2-1}-\frac{V_{0,l-n_1,0}}{n_2-1} \right)I_{[n_1+1,n_1+n_2]}(l) \nonumber \\
		& \qquad +{n_1\choose m_1}{n_2\choose m_2}{n_3-1\choose m_3} \left( \frac{n_3 U_n}{n_3-1}-\frac{V_{0,0,l-n_1-n_2}}{n_3-1} \right)I_{[n_1+n_2+1,n]}(l)
		\Bigg]\nonumber \\
		= & n U_{n}-\frac{n(n-1)}{n-m} U_n \Bigg[\frac{n_1-m_1}{n_1(n_1-1)} I_{[1,n_1]}(l) + \frac{n_2-m_2}{n_2(n_2-1)} I_{[n_1+1,n_1+n_2]}(l) \nonumber \\ & \qquad + \frac{n_3-m_3}{n_3(n_3-1)} I_{[n_1+n_2+1,n]}(l) \Bigg]
		+ \Bigg[\frac{n_1-m_1}{n_1(n_1-1)} V_{l,0,0} I_{[1,n_1]}(l) \nonumber \\ & \qquad + \frac{n_2-m_2}{n_2(n_2-1)} V_{0,l-n_1,0} I_{[n_1+1,n_1+n_2]}(l)
		+ \frac{n_3-m_3}{n_3(n_3-1)} V_{0,0,l-n_1-n_2} I_{[n_1+n_2+1,n]}(l) \Bigg]. \label{Vi}
	\end{align}
\end{small}
\noindent Therefore the values of $E V_l$ are as follows
\begin{footnotesize}
	\begin{align*}
		EV_l= & n \theta-\frac{n(n-1)}{n-m} \theta \Bigg[\frac{n_1-m_1}{n_1(n_1-1)} I_{[1,n_1]}(l) + \frac{n_2-m_2}{n_2(n_2-1)} I_{[n_1+1,n_1+n_2]}(l) + \frac{n_3-m_3}{n_3(n_3-1)} I_{[n_1+n_2+1,n]}(l) \Bigg]\\
		& \qquad + \Bigg[\frac{n_1-m_1}{n_1(n_1-1)} \theta I_{[1,n_1]}(l) + \frac{n_2-m_2}{n_2(n_2-1)} \theta I_{[n_1+1,n_1+n_2]}(l) + \frac{n_3-m_3}{n_3(n_3-1)} \theta I_{[n_1+n_2+1,n]}(l) \Bigg].
	\end{align*}
\end{footnotesize}
%By simplifying the above expression, we obtain
%\begin{equation*}
%	E V_l= \begin{cases}\theta\left(\frac{n}{n-m}\right)\left[n-m -\frac{(n-1)(n_1-m_1)}{n_1}\right], & \text { if }\; l=1,2, \ldots, n_{1}  \\~\\ \theta\left(\frac{n}{n-m}\right)\left[n-m -\frac{(n-1)(n_2-m_2)}{n_2}\right],  & \text { if } \;l=n_{1}+1, n_{1}+2, \ldots, n_{1}+n_{2} \\~\\
	%		\theta\left(\frac{n}{n-m}\right)\left[n-m -\frac{(n-1)(n_3-m_3)}{n_3}\right] , & \text { if }\; l=n_{1}+n_{2}+1, \ldots, n \end{cases}
%\end{equation*}

\vspace*{1mm}

\section*{\textbf{Appendix B}: \textbf{Proof of Theorem } \ref{three sample}}

We first provide some lemmas  which play a vital role in the proof of 
 Theorem \ref{three sample}.

% Here, without loss of generality, we assume $n_1<n_2<n_3$ and $m_1<<n_1$, $m_2<<n_2$ and $m_3<<n_3$.
\vspace*{2mm}

The following lemma ensures both the existence and uniqueness of the solution to equation \eqref{lambda}. Its proof is analogous to the one given by \cite{jing2009jackknife} for the two-sample case, hence omitted.
\vspace*{1mm}

\noindent \begin{lemma} \label{lemma2}
	Under conditions (a) and (b) of the Theorem \ref{three sample}, we have
	
	$$
	P\left\{\min _{1 \leq i \leq n}\left(V_{i}-E V_{i}\right)<0<\max _{1 \leq i \leq n}\left(V_{i}-E V_{i}\right)\right\} \longrightarrow 1\;\;\; \text{as }\; \min\{n_1,n_2,n_3\} \rightarrow \infty.
	$$
\end{lemma}
\noindent \begin{lemma} \label{lemma3}
	Suppose that conditions (a) and (b) of the Theorem \ref{three sample} hold and let $S_{n}^2=\frac{1}{n} \sum_{l=1}^{n}\left(V_{l}-E V_{l}\right)^{2}$. Then,
	
	$$
	S_{n}^2=n \sigma^{2}+o_{p}(1)\; \; \; \;\text{as }\; \min\{n_1,n_2,n_3\} \rightarrow \infty.
	$$
\end{lemma}
\begin{proof}
	Using \eqref{Vi}, for $1 \leq l \leq n_{1}$, we have
	
	\begin{align*}
		V_l=& \left( n- \frac{n(n-1)(n_1-m_1)}{(n-m)(n_1-1)} \right)U_n + \frac{n(n-1)(n_1-m_1)}{n_1(n-m)(n_1-1)}  V_{l,0,0} \end{align*}
	Therefore,
	\begin{align*}	V_{l}-E V_{l}=&\frac{n(n-1)(n_1-m_1)}{n_1(n-m)(n_1-1)}\left(V_{l, 0,0}-U_{n}\right) +\frac{n}{n-m} \left(n-m -\frac{(n-1)(n_1-m_1)}{n_1} \right)\left(U_{n}-\theta\right) \end{align*}
	and
	\begin{align*}
		\frac{1}{n_{1}} \sum_{l=1}^{n_{1}}\left(V_{l, 0,0}-U_{n}\right)\left(U_{n}-\theta\right)
		=  \left(U_{n}-\theta\right)\left( \frac{1}{n_1}\sum_{l=1}^{n_1} V_{l,0,0}-U_n \right) =0.
	\end{align*}
	Hence, we have that
	
	\begin{align*} \sum_{l=1}^{n_{1}} (V_{l}-E V_{l})^2=&\left[\frac{n(n-1)(n_1-m_1)}{n_1(n-m)(n_1-1)} \right]^2 \sum_{l=1}^{n_{1}}\left(V_{l, 0,0}-U_{n}\right)^2\\
		& \qquad+\left[\frac{n}{n-m} \left(n-m -\frac{(n-1)(n_1-m_1)}{n_1} \right)\right]^2 n_1\left(U_{n}-\theta\right)^2.
	\end{align*}
	Similarly, for $n_{1}<l \leq n_{1}+n_{2}$, we have
	\begin{align*} \sum_{l=n_1+1}^{n_{1}+n_2} (V_{l}-E V_{l})^2=&\left[\frac{n(n-1)(n_2-m_2)}{n_2(n-m)(n_2-1)} \right]^2 \sum_{l=n_1+1}^{n_{1}+n_2}\left(V_{0, l,0}-U_{n}\right)^2\\
		& \quad +\left[\frac{n}{n-m} \left(n-m -\frac{(n-1)(n_2-m_2)}{n_2} \right)\right]^2 n_2\left(U_{n}-\theta\right)^2.
	\end{align*}
	And  for $n_{1}+n_{2}<l \leq n$, we have
	\begin{align*} \sum_{l=n_1+n_2+1}^{n} (V_{l}-E V_{l})^2=&\left[\frac{n(n-1)(n_3-m_3)}{n_3(n-m)(n_3-1)} \right]^2 \sum_{l=n_1+n_2+1}^{n}\left(V_{0, 0,l}-U_{n}\right)^2\\
		&\qquad+\left[\frac{n}{n-m} \left(n-m -\frac{(n-1)(n_3-m_3)}{n_3} \right)\right]^2 n_3\left(U_{n}-\theta\right)^2.
	\end{align*}
	Therefore,
	
	\begin{small}
		\begin{align*}
			S_{n}^2=&\frac{1}{n} \sum_{l=1}^{n}\left(V_{l}-E V_{l}\right)^{2} \\
			= &\frac{1}{n} \left[\frac{n(n-1)}{n-m} \right]^2\bigg[\frac{(n_1-m_1)^2}{n_1(n_1-1)}\frac{1}{n_{1}(n_1-1)} \sum_{l=1}^{n_{1}}\left(V_{l, 0,0}-\bar{V}_{\cdot, 0,0}\right)^{2}\\
			&+\frac{(n_2-m_2)^2}{n_2(n_2-1)}\frac{1}{n_{2}(n_2-1)} \sum_{l=n_{1}+1}^{n_{1}+n_{2}}\left(V_{0, l, 0}-\bar{V}_{0,\cdot, 0}\right)^{2} +\frac{(n_3-m_3)^2}{n_3(n_3-1)}\frac{1}{n_{3}(n_3-1)}\left(V_{0,0, l}-\bar{V}_{0,0, \cdot}\right)^{2}\bigg]  \\
			& +\frac{1}{n}\left[\frac{n}{n-m}\right]^{2}\bigg[n_1\left(n-m -\frac{(n-1)(n_1-m_1)}{n_1} \right)^2+n_2 \left(n-m -n_2\frac{(n-2)(n_2-m_2)}{n_2} \right)^2\\
			& \qquad \qquad \qquad \qquad +n_3 \left(n-m -\frac{(n-1)(n_3-m_3)}{n_3} \right)^2\bigg]\left(U_{n}-\theta\right)^{2} . \tag{*}
		\end{align*}
	\end{small}
	
	From part (a) of Theorem \ref{lemma1}, we have $U_{n}-\theta=O_{p}\left(\min\{n_1,n_2,n_3\}^{-1 / 2}\right)$. Therefore, the second term in Equation (*) is
	\begin{align*}
		&	\frac{1}{n}\left[\frac{n}{n-m}\right]^{2}\bigg[n_1\left(n-m -\frac{(n-1)(n_1-m_1)}{n_1} \right)^2+n_2 \left(n-m -\frac{(n-1)(n_2-m_2)}{n_2} \right)^2\\
		& \qquad \qquad \qquad  +n_3 \left(n-m -\frac{(n-1)(n_3-m_3)}{n_3} \right)^2\bigg]\left(U_{n}-\theta\right)^{2} = O_p(\min\{n_1,n_2,n_3\}^{-1}).
	\end{align*}
	Moreover, the first term of equation (*) is (see Chapter 2 of \cite{wang2010empirical})
	\begin{small}
		\begin{align*}
			&\frac{1}{n} \left[\frac{n(n-1)}{n-m} \right]^2\bigg[\frac{(n_1-m_1)^2}{n_1(n_1-1)}\frac{1}{n_{1}(n_1-1)} \sum_{l=1}^{n_{1}}\left(V_{l, 0,0}-\bar{V}_{\cdot, 0,0}\right)^{2}\\
			&\quad+\frac{(n_2-m_2)^2}{n_2(n_2-1)}\frac{1}{n_{2}(n_2-1)} \sum_{l=n_{1}+1}^{n_{1}+n_{2}}\left(V_{0, l, 0}-\bar{V}_{0,\cdot, 0}\right)^{2} +\frac{(n_3-m_3)^2}{n_3(n_3-1)}\frac{1}{n_{3}(n_3-1)}\left(V_{0,0, l}-\bar{V}_{0,0, \cdot}\right)^{2}\bigg]  \\
			& =n \hat{\sigma}^{2}+o_{p}(1).
		\end{align*}
	\end{small}
	By substituting the above into equation (*) and using Theorem \ref{lemma1}, we obtain
	$$
	S_{n}^2=n \sigma^{2}+o_{p}(1).
	$$
	This prove the assertion.
\end{proof}
The proof of the following Lemma is straight forward and on the same lines as of \cite{an2018jackknife}, and hence omitted.
\vspace*{2mm}

\noindent \begin{lemma} \label{lemma4}
	Suppose that conditions (a) and (b) of the Theorem \ref{three sample} hold. Then, $$\max_{1 \leq l \leq n}\left|V_{l}-EV_l\right|=o_{p}\left(n^{1 / 2}\right) \;\; \text{ and }\;\; \frac{1}{n} \sum_{l=1}^{n}\left|V_{l}-EV_l\right|^{3}=o_{p}\left(n^{1 / 2}\right).$$
\end{lemma}
\noindent \begin{lemma} \label{lemma5}
	Under the conditions (a) and (b) of the Theorem \ref{three sample}, we have $$|\lambda|=O_{p}\left(n^{-1 / 2}\right) \;\; \text{ and }\;\; \lambda=\frac{\left(U-\theta\right)}{ S_{n}^2}+o_{p}\left(n^{-1 / 2}\right).$$
\end{lemma}
\noindent \begin{proof}
	Noting that, \eqref{lambda} together with the fundamental inequality $\mid x \pm$ $y|\geq| x|-| y \mid$ leads to

	\begin{align*}
		0 &=\left| 	\frac{1}{n} \sum_{i=1}^{n} \frac{V_{i}-EV_{i}}{1+\lambda\left(V_{i}-EV_{i}\right)}\right|  =\frac{1}{n}\left|\sum_{i=1}^{n}\left(V_{i}-E V_{i}\right)-\lambda \sum_{i=1}^{n} \frac{\left(V_{i}-E V_{i}\right)^{2}}{1+\lambda\left(V_{i}-E V_{i}\right)}\right| \\
		& \geq \frac{1}{n}\left|\lambda \sum_{i=1}^{n} \frac{\left(V_{i}-E V_{i}\right)^{2}}{1+\lambda\left(V_{i}-E V_{i}\right)}\right| -\frac{1}{n}\left|\sum_{i=1}^{n}\left(V_{i}-E V_{i}\right)\right| \\
		&\geq  \frac{|\lambda| S_{n}^2}{1+|\lambda| \left(\max_{1 \leq l \leq n}\left|V_{l}-EV_l\right|\right)}-\left|U_n-\theta\right|
	\end{align*}
	Using the Theorem \ref{lemma1}, the second term is  $O_{p}\left(n_{1}^{-1 / 2}\right)$ and, from Lemma \ref{lemma3}, we obtain $$|\lambda|\left(1+|\lambda| \left(\max_{1 \leq l \leq n}\left|V_{l}-EV_l\right|\right)\right)^{-1}=O_{p}\left(n^{-1 / 2}\right).$$
	Therefore, using $\max_{1 \leq l \leq n}\left|V_{l}-EV_l\right|=o_{p}\left(n^{1 / 2}\right)$ from Lemma \ref{lemma4}, we have
	\begin{equation} \label{modelambda}
		|\lambda|=O_{p}\left(n^{-1 / 2}\right).
	\end{equation}
	To prove that second part of lemma, let $\gamma_{l}=\lambda\left(V_{l}-E V_{l}\right),\;l=1,2,\dots,n$. Then, using Lemma \ref{lemma4}, we get

	\begin{align} \label{modegamma}
		\max _{1 \leq l \leq n}\left|\gamma_{l}\right| & =|\lambda| \max _{1 \leq l \leq n}\left|V_{l}-E V_{l}\right| =O_{p}\left(n^{-1 / 2}\right) o\left(n^{1 / 2}\right)=o_{p}(1).
	\end{align}
	From \eqref{lambda}, we have
	\begin{align*}
		0 &=	\frac{1}{n} \sum_{i=1}^{n} \frac{V_{i}-EV_{i}}{1+\lambda\left(V_{i}-EV_{i}\right)} \\ &=\frac{1}{n}\sum_{i=1}^{n}\left(V_{i}-E V_{i}\right)-\lambda \frac{1}{n}\sum_{i=1}^{n} \frac{\left(V_{i}-E V_{i}\right)^{2}}{1+\lambda\left(V_{i}-E V_{i}\right)} \\
		&=\frac{1}{n}\sum_{i=1}^{n}\left(V_{i}-E V_{i}\right)-\lambda \frac{1}{n}\sum_{i=1}^{n} \frac{\left(\left(V_{i}-E V_{i}\right)^{2}(1+\gamma_i)-\lambda\left(V_{i}-E V_{i}\right)^{3}\right)}{1+\gamma_i} \\
		& =U_n-\theta-\lambda S_{n}^2+\lambda^2\frac{1}{n} \sum_{i=1}^{n} \frac{\left(V_{i}-E V_{i}\right)^3 }{1+\gamma_{i}}.
	\end{align*}
	Also, using Lemma \ref{lemma4} and equations \eqref{modelambda}-\eqref{modegamma}, last term in the above expression is bounded by
	$$
	\frac{\frac{1}{n}\sum_{i=1}^{n}\left|V_{i}-E V_{i}\right|^{3}}{\left|1+\max _{1 \leq l \leq n}\left|\gamma_{l}\right|\right|} \lambda^{2}=o_p\left(n^{1 / 2}\right)o_{p}(1) O_{p}\left(n^{-1}\right) =o_{p}\left(n^{-1 / 2}\right)
	$$
	Therefore, we get
	$$
	\lambda=\frac{\left(U-\theta\right)}{ S_{n}^2}+o_{p}\left(n^{-1 / 2}\right).
	$$
	Hence the lemma is proved.
\end{proof}

%%%%%%%%%%%%%%%%%%%%%%%%%%%%%%%%%%%%%%%%%%%%%%%%%%%%%%%%%%%%%%%%%

\begin{proof} \textbf{of Theorem \ref{three sample}: }
	Now, we are ready to prove the theorem. Take $\gamma_{l}=\lambda\left(V_{l}-EV_l\right),\; l=1,2,\dots,n,$ and consider the $\log$ likelihood ratio given by

	\begin{align*}
		-2 \log R(\theta) & =2 \sum_{1=1}^{n} \log \left(1+\lambda\left(V_{i}-EV_i\right)\right) \\
		& =2 \sum_{1=1}^{n} \log \left(1+\gamma_{i}\right) .
	\end{align*}

	From the Taylor series expansion for $\log \left(1+\gamma_{i}\right)$, we obtain
	\begin{equation}\label{taylor}
		-2 \log R(\theta)=2 \sum_{1=1}^{n} \gamma_{i}-\sum_{1=1}^{n}{\gamma_{i}}^{2}+\frac{2}{3} \sum_{1=1}^{n} \gamma_{i}{ }^{3}-\frac{1}{2} \sum_{1=1}^{n}{\gamma_{i}}^{4}+\ldots
	\end{equation}
	
	Let us first consider the third term of \eqref{taylor}
	\begin{align*}
		\frac{2}{3} \sum_{1=1}^{n} \gamma_{i}{ }^{3} =\frac{2}{3} \sum_{1=1}^{n} \lambda^{3}\left(V_{i}-EV_i\right)^{3}
		& \leq \frac{2}{3}\, n^{\frac{3}{2}}|\lambda|^3\frac{1}{n} \sum_{1=1}^{n}\left|V_{i}-EV_i\right|^{3} \frac{1}{\sqrt{n}} \\
		& =n^{\frac{3}{2}} O_{p}(n^{-3/2}) o_{p}(n^{\frac{1}{2}}) o_{p}(n^{-\frac{1}{2}}) \tag{By using Lemmas \ref{lemma4} and \ref{lemma5}}\\
		& = o_{p}(1).
	\end{align*}
	Hence, it is straightforward to verify that all terms with powers of $\gamma_i$ equal to or greater than 3 will be $o_{p}(1)$.\\
	
	\noindent Next, we consider the second term of \eqref{taylor} as
	\begin{align*}
		\sum_{1=1}^{n} \gamma_{i}^{2}  =\sum_{1=1}^{n} \lambda^{2}\left(V_{i}-\theta\right)^{2}
		= n\lambda^2 S_n^2
		& \leq n|\lambda|^2 (n S_{n_1,n_2,n_3}^2+o_p(1)) \tag{By using Lemma \ref{lemma3}} \\
		& = n\; o_p(n^{-1}) (n S_{n_1,n_2,n_3}^2+o_p(1)) \tag{By using Lemma \ref{lemma5}}\\
		& = o_{p}(1).
	\end{align*}
	Now, consider the first term of \eqref{taylor}
	\begin{align*}
		2 \sum_{1=1}^{n} \gamma_{i}  =2 \sum_{1=1}^{n} \lambda\left(V_{i}-EV_i\right)
		& = 2\lambda n \frac{1}{n} \sum_{1=1}^{n}\left(V_{i}-EV_i\right) \\
		& = 2\lambda n(U_n-\theta) \\
		& =\frac{n(U_n-\theta)^{2}}{S_n^2} + n o_p(n^{-1/2})(U_n-\theta)\tag{By using Lemma \ref{lemma5}}\\
		& =\frac{n(U_n-\theta)^{2}}{S_n^2} + n\, o_p(n^{-1/2})O_p(n^{-1/2})\tag{By using Theorem \ref{lemma1}}\\
		& =\frac{n(U_n-\theta)^{2}}{S_n^2} + o_p(1).
	\end{align*}\\
	Therefore,
	\begin{align*}
		-2 \log R(\theta)  =\frac{n(U_n-\theta)^{2}}{S_n^2} +o_p(1)
		=\frac{(U_n-\theta)^{2}}{\sigma^{2}} +o_p(1)	\; \; \; \;\text{as }\; n_{1} \rightarrow \infty. \tag{By using Lemma \ref{lemma3}}
	\end{align*}
	
	By using part (b) of Theorem \ref{lemma1} and the Slutsky's theorem, $-2 \log R(\theta)$ converges in distribution to a $\chi^{2}$ random variable with one degree of freedom. Hence the theorem is proved.\\
\end{proof}

\end{document}